\documentclass{article}
\usepackage{dcolumn}
\usepackage{bm}
\usepackage{verbatim}       

\usepackage[dvips]{graphicx}
\usepackage{amsthm}
\usepackage{amssymb}
\usepackage{bm}
\usepackage{amstext}
\usepackage{psfrag}
\usepackage{amsmath}
\usepackage{amsfonts}

\newfont{\bb}{msbm10 at 12pt}

\newcommand{\p}{\partial}

\newcommand{\bd}{\begin{definition}}                
\newcommand{\ed}{\end{definition}}                  
\newcommand{\bc}{\begin{corollary}}                 
\newcommand{\ec}{\end{corollary}}                   
\newcommand{\bl}{\begin{lemma}}                     
\newcommand{\el}{\end{lemma}}                       
\newcommand{\bp}{\begin{proposition}}            
\newcommand{\ep}{\end{proposition}}                
\newcommand{\bere}{\begin{remark}}                  
\newcommand{\ere}{\end{remark}}                     

\newcommand{\bt}{\begin{theorem}}
\newcommand{\et}{\end{theorem}}

\newcommand{\be}{\begin{equation}}
\newcommand{\ee}{\end{equation}}

\newcommand{\bit}{\begin{itemize}}
\newcommand{\eit}{\end{itemize}}
\newtheorem{theorem}{Theorem}[section]
\newtheorem{corollary}[theorem]{Corollary}

\newtheorem{lemma}[theorem]{Lemma}
\newtheorem{proposition}[theorem]{Proposition}
\theoremstyle{definition}
\newtheorem{definition}[theorem]{Definition}
\theoremstyle{remark}
\newtheorem{remark}[theorem]{Remark}
\newtheorem{example}[theorem]{Example}

\begin{document}
%

\title{Globally hyperbolic spacetimes can be defined without the `causal' condition}

\author{R. A. Hounnonkpe\footnote{Universit\'e d'Abomey-Calavi, B\'enin and Institut de Math\'ematiques et de Sciences Physiques (IMSP), Porto-Novo, B\'enin. E-mail:
rhounnonkpe@ymail.com} \ and \
E. Minguzzi\footnote{Dipartimento di Matematica e Informatica ``U. Dini'', Universit\`a degli Studi di Firenze,  Via
S. Marta 3,  I-50139 Firenze, Italy. E-mail:
ettore.minguzzi@unifi.it}}

\date{}
\maketitle

\begin{abstract}
\noindent Reasonable spacetimes are non-compact and of dimension larger than two. We show that  these spacetimes are globally hyperbolic if and only if the causal diamonds are compact. That is, there is no need to impose the causality condition, as it can be deduced.
We also improve the definition of global hyperbolicity for the non-regular theory (non $C^{1,1}$ metric) and for general cone structures by proving  the following convenient characterization for upper semi-continuous cone distributions: causality and the causally convex hull of compact sets is compact. In this case the causality condition cannot be dropped, independently of the spacetime dimension.
 Similar results are obtained for causal simplicity.
\end{abstract}



\section{Introduction}
The strongest causality condition in the causal hierarchy of spacetimes is global hyperbolicity, the next property being  causal simplicity \cite{minguzzi18b}.
This work is devoted to the simplification of their  definitions.


The causality condition of {\em global hyperbolicity} is of course the best known, for its fundamental importance in general relativity is widely recognized. Suffice here to mention   the Choquet-Bruhat and Geroch's theorem on the initial value problem \cite{choquet69} \cite[Thm.\ 10.2.2]{wald84}, the strong cosmic censorship conjecture \cite{penrose69}, and the metric splitting \cite{bernal03}.

Though less known {\em causal simplicity} has also attracted considerable interest, particularly in the study of the space of lightlike geodesics \cite{low89,chernov18,hedicke19}. It plays an important role  in questions of geodesic connectedness (e.g.\ in dimensional reduction of type: spacelike \cite[Thm.\ 3.13]{minguzzi06f}, lightlike \cite{minguzzi06d} or timelike \cite{caponio10,harris15})  and is  related to the notion of Cauchy holed spacetime \cite{geroch77,krasnikov09,manchak09,minguzzi12c}.

Interest on these properties has never faded so it is a bit surprising that new foundational results  can  still be obtained. This fact proves better than nothing else the richness and fruitfulness  of these concepts.

In Lorentzian causality theory the traditional definition of global hyperbolicity is the following \cite{hawking73} (the first definition by Leray \cite{leray52,hawking73} in terms of the space connecting causal curves is rarely used in causality theory and will not be recalled):
\begin{definition} \label{ngo}
A spacetime is globally hyperbolic if
\begin{itemize}
\item[(a)] it is strongly causal,
\item[(b)] the causal diamonds are compact.
\end{itemize}
\end{definition}
Similarly, the traditional definition of causal simplicity is \cite{hawking73}
\begin{definition} \label{klo}
A spacetime is causally simple if
\begin{itemize}
\item[($\alpha$)] it is distinguishing,
\item[($\beta$)] for every $p\in M$, $J^+(p)$ and $J^-(p)$ are closed.
\end{itemize}
\end{definition}
The reader is referred to \cite{hawking73,minguzzi18b} for the definitions of the strongly causal and distinction properties. We shall not use them because   Bernal and S\'anchez \cite{bernal06b} have shown that (a) and ($\alpha$) can be both weakened to {\em causality}: there is no closed causal curve.

It turns out that in several cases  better improvements are possible. Before we describe them let us recall our notations and terminologies.
In this work a manifold is always Hausdorff and second countable, hence paracompact. A spacetime $(M,g)$ is a connected time-oriented Lorentzian manifold whose metric $g$ is $C^2$ ($C^{1,1}$ will be enough) and of signature $(-,+,\cdots, +)$.
The inclusion $\subset$ is reflexive. With a curve $\gamma$ we might mean a map  $\gamma\colon I \to M$ or the image of the map.
We write $p<q$ if there is a causal curve connecting $p$ to $q$, and $p\ll q$ if there is a timelike curve connecting $p$ to $q$. We write $p\le q$ if $p<q$ or $p=q$. The sets $J=\{(p,q)\colon p\le q\}$ and $I=\{(p,q)\colon p\ll q\}$ are the causal and chronological relations respectively.
A {\em causal diamond} is a set of the form $J^+(p)\cap J^-(q)$.
The {\em causally convex hull} of a set $S\subset M$, is the set $J^+(S)\cap J^-(S)$, namely the union of the images of all the causal curves which start and end in $S$.
For most results of causality theory we refer to the recent review \cite{minguzzi18b}.

\section{Improved definitions}

Let us start by recalling a result of causality theory \cite[Thm.\ 4.12]{minguzzi18b}.
\begin{theorem} \label{clo}
The following properties are equivalent:
\begin{itemize}
\item[(i)] $J$ is closed in the topology of $M\times M$,
\item[(ii)] $J^+(p)$ and $J^-(p)$ are closed for every $p\in M$,
\item[(iii)] $J^+(K)$ and $J^-(K)$ are closed for every compact subset $K\subset M$,
\item[(iv)] the causal diamonds are closed.
\end{itemize}
\end{theorem}
Thus in the definition of causal simplicity any of these properties can be conveniently used.

Our first contribution is to include another useful property to the previous list, namely
\begin{itemize}
\item[$(v)$] {\em the causally convex hulls of compact sets are closed}.
\end{itemize}

This result is a bit surprising because we are not demanding causality, a fact that will prove  important, as we shall see.
It will follow from the more general Prop.\ \ref{kcp}.

\begin{lemma} \label{hbp}
Let $(M,g)$ be a spacetime, and let $K=\{p,q\}$ with $p,q\in M$, then
\begin{align*}
J^+(K)\cap J^-(K)&=J^+(p)\cap J^-(q),  &\textrm{if} \ p\le q,\\
J^+(K)\cap J^-(K)&=J^+(q)\cap J^-(p),  &\textrm{if} \ q\le p, \\
J^+(K)\cap J^-(K)&=[J^+(p)\cap J^-(p)]\cup [J^+(q)\cap J^-(q)],  &\textrm{otherwise}.
\end{align*}
\end{lemma}

\begin{proof}
If $p\le q$, then $J^+(p)\cap J^-(p)\subset J^+(p)\cap J^-(q)$,  $J^+(q)\cap J^-(q)\subset J^+(p)\cap J^-(q)$, thus
\begin{align*}
J^+(K)\cap J^-(K)&=[J^+(p)\cap J^-(p)] \cup [J^+(q)\cap J^-(q)]\\
& \quad \cup [J^+(p)\cap J^-(q)]\cup [J^+(q)\cap J^-(p)]\\
&=[J^+(p)\cap J^-(q)]\cup [J^+(q)\cap J^-(p)]
\end{align*}
but if $J^+(q)\cap J^-(p)\ne \emptyset$ then $q\le p$, thus $J^\pm(q)=J^\pm(p)$, hence $J^+(q)\cap J^-(p)=J^+(p)\cap J^-(q)$. Thus in any case we have
\[
J^+(K)\cap J^-(K)=J^+(p)\cap J^-(q).
\]
The other cases are analogous or trivial.
\end{proof}

\begin{proposition} \label{kcp}
The closedness (compactness) of the causal diamonds is equivalent to the closedness (resp.\ compactness) of the causally convex hulls of compact sets.
\end{proposition}

\begin{proof}

$\Rightarrow$, closedness case.

With Thm.\ \ref{clo} we have seen that the closedness of the causal diamonds implies that for every compact set $K$, $J^+(K)$ and $J^-(K)$ are closed, and hence that $J^+(K)\cap J^-(K)$ is closed.

$\Rightarrow$, compactness case.

Let $K$ be a compact set and let
$\tilde{K}$
be a compact set such that $K \subset  \textrm{Int} \tilde{K}$ . For each $p \in K$
we can find $q, r \in  \tilde{K}$ such that $p \in I^+(q)\cap I^-(r)$, so we can find a finite covering of
$K$, given by sets of the form $I^+(q_i) \cap I^-(r_i)$, $i = 1, \cdots, s$. Then $J^+(K) \cap J^-(K) \subset
\cup_{i,j}J^+(q_i) \cap J^-(r_j)$, which being the union of compact sets is compact. Since $J^+(K) \cap J^-(K)$ is a closed subset of a compact set it is compact.

$\Leftarrow$. In both the closedness and compactness cases it is a trivial consequence of  Lemma \ref{hbp}.
\end{proof}

For completeness, let us now recall the up to date definition of global hyperbolicity \cite[Def.\ 4.117]{minguzzi18b}.
The reader is once again referred to \cite{minguzzi18b} for the various terms entering this definition. We do not insist on them because for our work we shall focus on (1).

\begin{definition} \label{sot}
A spacetime is {\em globally hyperbolic} if the following equivalent conditions hold
\begin{itemize}
\item[(1)] Causality and for every $p,q\in M$, $J^{+}(p)\cap J^-(q)$ is compact.
\item[(2)] Non-total imprisonment and for every $p,q\in M$, $\overline{J^{+}(p)\cap J^-(q)}$ is compact.
\item[(3)] Stable causality and for every $p,q\in M$,  $J_S^{+}(p)\cap J_S^-(q)$ is  compact.
\item[(4)] There exists a Cauchy hypersurface.
\end{itemize}
\end{definition}

As mentioned, the improved
version (1) appeared in  \cite{bernal06b} and it is particularly convenient for the inclusion of global hyperbolicity into the causal ladder.

The version (2) appeared in
\cite{minguzzi08e}. It has several advantages over the others because it does not
require that the causal diamonds be closed (after all this property is so strong that, as we want to prove, it implies that in most cases the `causal' condition can be dropped).
Characterization (2) makes it clear
that by narrowing the cones one does not spoil global hyperbolicity, a fact
not at all obvious from the other formulations, save for (4). In fact trough this version it is possible to give a simple proof of the stability of global hyperbolicity \cite[Thm.\ 2.39]{minguzzi17}.

The
characterization (3) appeared in \cite{minguzzi17} and it is interesting because it shows that global hyperbolicity
is expressible through the Seifert relation (stable causality is equivalent to its antisymmetry).

Finally, the last characterization (4) due to Geroch is classical \cite{geroch70,hawking73}. It implies that the spacetime splits topologically as $M\sim S\times \mathbb{R}$ where $S$ is the Cauchy hypersurface. Bernal and S\'anchez have shown that in (4) the Cauchy hypersurface can be demanded to be $C^1$ and spacelike and that the splitting  can be chosen to be a diffeomorphism so that the  metric splits at each point \cite{bernal03}, see also \cite{fathi12,chrusciel13,muller13,bernard16,samann16,minguzzi17} for proofs holding under  weaker differentiability conditions or for more general cone distributions.

A first generalization that follows immediately from Prop.\ \ref{kcp} is

\begin{proposition} \label{mbj}
A spacetime is globally hyperbolic iff it is causal and the causally convex hulls of compact sets are compact.

A spacetime is causally simple iff it is causal and the causally convex hulls of compact sets are closed.
\end{proposition}

In the next section  we show that {\em causal} can be dropped from these statements in most cases.

\subsection{Spacetime dimension larger than 2}

Reasonable spacetimes have, of course, dimension four. In  this section we explore the consequences of the condition $n+1\ge 3$, where $n+1$ is the spacetime dimension. It must also be said that this assumption is not unusual even in purely mathematical works, for 2-dimensional spacetimes are very peculiar in that the geodesics do not admit conjugate points.

We recall that a spacetime is {\em chronological} if it does not contain closed timelike curves. On a spacetime the set  $\mathcal{C}$ of points through which there passes a closed timelike curve is called {\em chronology violating set}.

A spacetime is {\em non-totally vicious} if there is some point through which no closed timelike curve passes. Stated in another way, the chronology violating set does not coincide with $M$, $\mathcal{C}\ne M$. Most authors in mathematical relativity agree that this is a very reasonable condition. Researchers do not study totally vicious spacetimes because their causality is trivial $J=I=M\times M$.

A spacetime is {\em reflecting} if the causal relation  satisfies the property $p\in \overline{J^-(q)}\Leftrightarrow q\in \overline{ J^+(p)}$. Clearly, if $J$ is closed then $J^\pm(r)$ is closed for every $r$ (by Thm. \ref{clo}) and so the spacetime is reflecting.

We start by recalling a result due to Clarke and Joshi \cite{clarke88} \cite[Prop.\ 4.26]{minguzzi18b}

\begin{proposition} \label{iqa}
A reflecting non-totally vicious spacetime is chronological.
\end{proposition}

In this section we are going to show that with the dimension condition the causality condition in the definition of {\em global hyperbolicity} and {\em causal simplicity} can be weakened to chronology, and then, thanks to this result, to {\em non-totally vicious}.

\begin{theorem} \label{jia}
Let $(M,g)$ be a spacetime of dimension $n+1\ge 3$ which is non-totally vicious. In the definitions \ref{ngo}-\ref{klo} of global hyperbolicity and causal simplicity the causality conditions (a) and ($\alpha$) can be dropped.
\end{theorem}
As shown in the previous section the other condition can be replaced by: the causally convex hull of compact sets is compact (closed, for causal simplicity); or equivalently by: the causal diamonds are compact (closed, for causal simplicity).

\begin{proof}
The condition that the causal diamonds are closed implies that $J$ is closed, hence that the spacetime is reflecting and hence, by Clarke and Joshi result, that it is chronological. We want to prove that it is causal. Suppose not, then there is a closed achronal lightlike geodesic (it is  $C^3$ because the connection is $C^1$) $\gamma\colon [a,b]\to M$, $\gamma(a)=\gamma(b)$, $\dot \gamma(a)\propto \dot \gamma(b)$, see e.g.\ \cite[Prop.\ 4.32]{minguzzi18b}
 (achronality implies that there are no discontinuities in the direction of the tangent vector).

Let $p\in \gamma$. Let $r\in \gamma$, as $\gamma$ is closed $r \in J^+(p)$. For every $s\in I^+(r)$ we have $s\in I^+(p)$, so by the arbitrariness of $s$,  $I^+(r)\subset I^+(p)$. Inverting the roles of $p$ and $r$ and considering the time dual result we get
 $I^+(r)=I^+ (p)$ and $I^-(r)=I^-(p)$.
The set $N:=\p I^+(\gamma)=\p I^+(p)$, which is necessarily non-empty as it contains $\gamma$, is an achronal boundary hence a $C^0$ achronal hypersurface \cite[Thm.\ 2.87]{minguzzi18b}.

Let $q\in \p I^+(p)\backslash \gamma$. It exists because the $C^1$ map $\gamma\colon[0,1]\to N$ has a domain consisting entirely of critical points so  by the Morse-Sard theorem \cite{hirsch76}  its image
cannot fill a manifold of dimension $n=\textrm{dim} N\ge 2$. Alternatively, the result follows from the fact that locally Lipschitz maps do not increase the Hausdorff dimension \cite[Thm.\ 5.5]{mattila95} \cite[Sec.\ 2.4.1]{evans92}.

Since $J$ is closed, $q\in J^+(p)$ but the causal curve connecting $p$ to $q$ must be a lightlike geodesic, for otherwise $q\in I^+(p)$, a contradiction. However, it cannot be aligned with $\gamma$ otherwise $q\in \gamma$. So if we take $r\in \gamma$, $r\ne p$, then $q\in I^+(r)=I^+(p)$, because there is a corner in the piecewise lightlike geodesic connecting $r$ to $q$,   again a contradiction.
\end{proof}

The spacetimes of  the previous theorem are certainly reasonable. Nevertheless, the reader can  prefer the following version where the non-totally vicious property is replaced by non-compactness.

\begin{theorem}
A non-compact spacetime of dimension $n+1\ge 3$  is globally hyperbolic iff its causal diamonds are compact.
\end{theorem}
Again the last condition can be replaced with: the causally convex hull of compact sets is compact. As we can see, all causality conditions connected to the formation of closed causal curves have been removed. The non-compactness condition is also very reasonable and most physicists would not  hesitate in including it  in the very definition of spacetime.

\begin{proof}
If the spacetime is totally vicious then for every $p\in M$ we have $[p]:=I^+(p)\cap I^-(p)=M$, see \cite[Prop.\ 4.20]{minguzzi18b}, and hence $I^\pm(p)=J^\pm(p)=M$. But then $M=J^+(p)\cap J^-(p)$ being a causal diamond must be compact, a contradiction.
\end{proof}

\subsection{Low differentiability and general cone distributions}

Recently there has been considerable interest in causality theory under weak differentiability conditions and for general cone distributions, see for instance \cite{minguzzi13d,kunzinger13b} for the $C^{1,1}$ theory, \cite{chrusciel12,sbierski15,samann16,galloway18b} for the $C^0$ metric theory, and \cite{fathi12,bernard16,minguzzi17} for  the general cone distribution case. It is natural to ask if the results of this work extend to the more general frameworks.

In this connection it is particularly convenient to introduce the notion of {\em closed cone structure} $(M,C)$, cf.\ \cite{bernard16,minguzzi17}, which is a distribution of closed sharp convex non-empty cones $p\mapsto C_p\subset T_pM\backslash 0$ such that $C=\cup_p C_p$ is closed in the topology of the slit tangent bundle $TM\backslash 0$. This last property is equivalent to the upper semi-continuity of the distribution $p\mapsto C_p$, see \cite[Prop.\ 2.3]{minguzzi17}. This structure preserves many results of causality theory  including the validity of the causal ladder  \cite[Thm.\ 2.47]{minguzzi17}. Notice that it is so general that it includes distributions of half-lines determined by continuous non-vanishing vector fields.

In a closed cone structure one has a well defined notion of causal relation, but no useful notion of chronological relation.
A sensible notion of chronological relation is obtained in a {\em proper cone structure} which is a closed cone structure for which  $(\textrm{Int} C)_p\ne \emptyset$ for every $p\in M$. Still, in these structures the identity $\p I^+(p)=\p J^+(p)$ or the identity $I\circ J\cup J\circ I \subset I$ do not necessarily hold. Spacetimes with $C^0$ metrics provide the simplest examples of proper cone structures.

\subsubsection{Proper cone structure}
Let us check if the results of the previous sections pass to the proper cone structure case.
The equivalence of properties $(i)-(iv)$ in Thm.\ \ref{clo} still holds true, see  \cite[Prop.\ 2.19, Thm.\ 2.37, Lemma 2.5]{minguzzi17}.
The proofs of Lemma \ref{hbp} and Prop.\ \ref{kcp} and hence the inclusion of property $(v)$ into the equivalence, passes through to the proper cone structure case with no alteration. The definition of causal simplicity remains the same (causality and closedness of $J$). For a proper cone structure the equivalence between (1)-(4) in the definition of global hyperbolicity still holds true \cite[Cor.\ 2.4, Thm.\ 2.45]{minguzzi17}, so Prop.\ \ref{mbj} passes through with no alteration in proof. For a proper cone structure the notion of reflectivity is also unaltered \cite[Prop.\ 2.16]{minguzzi17}, as it is the definition of non-totally vicious spacetime and the proof of Clarke and Joshi result (Prop.\ \ref{iqa}).

Unfortunately, we do not know if Thm.\ \ref{jia} passes to the proper cone structure case, certainly its proof does not.
Here $\gamma$ can be treated as an achronal continuous causal curve, but in the very last step of the proof we used the fact that these curves cannot branch. Unfortunately, we do not have at our disposal a non-branching result of this type (save for the $C^{1,1}$ Lorentz-Finsler theory \cite{minguzzi13d} for which the theorem does indeed hold).

\subsubsection{Closed cone structure}

In the closed cone structure case  the definition of causal simplicity is the usual one (causality and closedness of $J$), but Thm.\ \ref{clo} fails in some directions (the equivalence between $(i)$ and $(iii)$ holds true \cite[Prop.\ 1.4]{nachbin65}\cite[Thm.\ 2.37]{minguzzi17}).

\begin{example}
Let $M=\mathbb{R}^2\backslash\{(0,0)\}$ endowed with the translational invariant cone distribution $C_p=\mathbb{R}_+ \p_y$. Here properties $(ii)$, $(iv)$ of Thm.\ \ref{clo} hold true but $(i)$, $(iii)$ and $(v)$ do not.
\end{example}

As for global hyperbolicity,  its definition requires some care \cite{bernard16} \cite[Def,\ 2.20]{minguzzi17}.
 Definitions (3) and (4) of Def.\ \ref{sot} are valid and equivalent in the closed cone structure case \cite[Def.\ 2.20, Thm.\ 2.45]{minguzzi17},
but, as shown by  the previous example, the formulations of (1) and (2) have to be strengthened. Condition (2) becomes:
\begin{itemize}
\item[(2')] Non-imprisonment and the causally convex hulls of relatively compact sets are relatively compact.
\end{itemize}
As for (1) the following have been shown to be valid formulations  \cite[Def.\ 2.20]{minguzzi17} equivalent to (3) and (4),
\begin{itemize}
\item[(1a)] Causal simplicity and the causally convex hulls of compact sets are compact \cite{minguzzi12d,minguzzi17}.
\item[(1b)] Causality and for every pair of  compact sets $K_1,K_2$ the `causal emerald' $J^+(K_1)\cap J^-(K_2)$ is compact \cite{bernard16}.
\end{itemize}
The condition of compactness of causal emeralds implies the closedness of $J$, \cite[Thm.\ 2.38]{minguzzi17} \cite{bernard16}. Still these definitions suggest that the next more elegant formulation might be true
\begin{itemize}
\item[(1c)] Causality and the causally convex hulls of compact sets are compact.
\end{itemize}
This would be rather pleasing because in Prop.\ \ref{mbj} we proved that this is an acceptable definition in the $C^2$ metric case.
It is clear that (1c) improves both (1a) and (1b).

\begin{theorem} \label{mbi}
For a closed cone structure and under causality the property `causally convex hulls of compact sets are closed'  is equivalent to the closedness of $J$.
\end{theorem}

\begin{proof}
$\Leftarrow$.
Immediate from the equivalence between  $(i)$ and $(iii)$ as proved in \cite[Prop.\ 1.4]{nachbin65}\cite[Thm.\ 2.37]{minguzzi17}

$\Rightarrow$. First let us prove that for every $p$ the sets $J^+(p)$ and $J^-(p)$ are closed. Let $q\in \overline{J^+(p)}$ and let $q_n\to q$ with $p\le q_n$. We want to prove that $p\le q$, which is obvious for $q=p$, so we can restrict ourselves to the case $q\ne p$. The set $K=\{p,q,q_1,q_2,\cdots\}$ is compact, thus $\hat K:=J^+(K)\cap J^-(K)$ is closed. Let $\sigma_n$ be a continuous causal curve connecting $p$ to $q_n$. We have $\sigma_n\subset \hat K$. By the limit curve theorem \cite[Thm.\ 2.14]{minguzzi17} either there is a limit continuous causal curve connecting  $p$ to $q$, hence $p\le q$, and we have finished, or there is a past inextendible limit continuous causal curve $\sigma^q$ ending at $q$. Since $\hat K$ is closed, we have $\sigma^q\subset \hat K$. Let $r\in \sigma^q\backslash\{q\}$, then $r\in \hat K$, thus $r\in J^+(K)$. There are various possibilities. Since $r < q$ the case $q\le r$ contradicts causality, thus it does not apply. The case $q_n\le r$ for some $n$ gives $p\le q$, because $r\le q$ and $p\le q_n$, and we have finished. The case $p\le r$ clearly gives $p\le q$. We conclude that  $J^+(p)$ is closed.

Now let $(p_n,q_n)\to (p,q)$, where $p_n\le q_n$, and let $\sigma_n$ be a continuous causal curve connecting $p_n$ to $q_n$.
We want to prove that $p\le q$, so we can assume $p\ne q$, for in the equality case this fact is trivial.
We also assume that $p_n\nleq q$ for all but a finite number of values of  $n$ otherwise by the closure of $J^-(q)$, $p \le q$ and we have finished. Thus we can assume, by passing to a subsequence if necessary, that $p_n\nleq q$ for every $n$.

As before by the limit curve theorem either  there is a limit continuous causal curve connecting  $p$ to $q$, in which case we have finished, or there is a past inextendible limit continuous causal curve $\sigma^q$ ending at $q$. The set $K=\{p,p_1, p_2, \cdots\} \cup \{q,q_1,q_2,\cdots\}$ is compact, thus $\hat K:=J^+(K)\cap J^-(K)$ is closed. Since $\sigma_n\subset \hat K$, we have $\sigma^q\in \hat K\subset J^+(K)$. Let $r\in\sigma^q\backslash\{q\}$, so that $r\in J^+(K)$. If $q_n\le r$ for infinitely many values of $n$, then by the closure of $J^-(r)$, $q\le r$ which is not possible since as $r< q$, it would violate causality. Similarly, if $q\le r$.
Thus we can assume that $p\le r$ or $p_n\le r$ for some $n$. The latter passibility is excluded because it would imply $p_n\le q$. The former case gives $p\le r\le q$, and we have finished.
\end{proof}

As  a consequence we have the next hoped for generalization of Prop.\ \ref{mbj}

\begin{corollary}
Let $(M,C)$ be a closed cone structure.

A spacetime is globally hyperbolic iff it is causal and the causally convex hulls of compact sets are compact.

A spacetime is causally simple iff it is causal and the causally convex hulls of compact sets are closed.

The characterizations (1c), (2'), (3) and (4) of global hyperbolicity are all equivalent.
\end{corollary}

\begin{proof}
The second statement follows immediately from Thm.\ \ref{mbi}. The first statement follows from the second and characterization (1a). The last statement is just a summary of already obtained results.
\end{proof}

Interestingly, the equivalence established in Theorem \ref{mbi}  does not hold without the causality assumption.

\begin{example}
The example has topology $S^1\times \mathbb{R}\times \mathbb{R}$, but the middle factor could be replaced with $\mathbb{R}^k$, $k\ge 0$, so showing that it does not depend on dimension. We start from the manifold  $N=\mathbb{R}\times \mathbb{R}\times \mathbb{R}$ and consider the vector field $v(x,y,z)=(1,0,1-z^2)$, for $z\in (-1,1)$, $v=(1,0,0)$ for $z=-1$, and $v(x,y,z)=v(x,y,z+2)$. The periodicity in $z$ implies $v=(1,0,0)$ for $z=1$. Notice that $v$ has integral lines of the form $t\mapsto (t+c_1,c_2,\tanh t)$ in $z\in (-1,1)$, where $c_1,c_2$, are constants, so they are asymptotic to the slices $z=-1$ and $z=1$, which are themselves generated by integral curves.
Now identify $x=0$ and $x=1$.
The closed cone structure given by $C=\mathbb{R}_+ v$ is not causal but it is such that the causally convex hulls of compact sets are compact. Moreover, taking $p$ such that $z(p)=-1$ and $q$ such that $z(q)=+1$, $y(q)=y(p)$, we have $(p,q)\in \bar J\backslash J$.
\end{example}

\section{Conclusions}

The very definition of spacetime is to some degree a matter of convention.
Most physicists  would concede that they could be defined as connected time-oriented Lorentzian manifolds of dimension four and non-compact or non-totally vicious, for otherwise their causal structure would be  trivial. We have shown that they are globally hyperbolic iff the causal diamonds are compact, with no need to introduce a ``causal'' condition, and hence with a considerable simplification of the traditional definition. Similar results have been obtained for the notion of causal simplicity.

Finally, the study of the regular and non-regular cases has clarified the importance of the condition `the  causally convex hull operation  preserves compactness'. In the regular case it can replace the standard assumption of compactness of the causal diamonds. For closed cone structures, it provides a characterization of global hyperbolicity but only jointly with causality.


\end{document}